\documentclass[letterpaper, 10 pt, conference]{ieeeconf}  

\IEEEoverridecommandlockouts                              

\usepackage{graphicx}
\usepackage{amsmath}
\usepackage{amssymb}
\usepackage{dsfont} 
\usepackage{tikz,pgfplots}
\pgfplotsset{compat=newest} 
\pgfplotsset{plot coordinates/math parser=false} 

\DeclareMathOperator{\Span}{span}
\newtheorem{thm}{Theorem}
\newtheorem{lem}{Lemma}
\newtheorem{prop}{Proposition}

\title{\LARGE \bf
Geometric stability considerations of the ribosome flow model with pool$^*$
}

\author{Wolfgang Halter$^{1}$, Jan Maximilian Montenbruck$^{1}$ and Frank Allg\"ower$^{1}$
\thanks{*Supported by the research cluster $\text{BW}^2$ (www.bwbiosyn.de) of the Ministry for Science, Research and Art Baden-W\"urttemberg}
\thanks{$^{1}$Authors are with Institute for Systems Theory and Automatic Control, 70569 Stuttgart, Germany.
        {\tt\small E-Mail of corresponding author: halter@ist.uni-stuttgart.de}.}%
}

\begin{document}

\maketitle
\thispagestyle{empty}
\pagestyle{empty}

\begin{abstract}                
In order to better understand the process of gene translation, the ribosome flow model (RFM) with pool was introduced recently. This model describes the movement of several ribosomes along an mRNA template and simultaneously captures the dynamics of the finite pool of ribosomes. Studying this system with respect to the number and stability of its equilibria was so far based on monotone systems theory \cite{Margaliot2012}. We extend the results obtained therein by using a geometric approach, showing that the equilibria of the system constitute a normally hyperbolic invariant submanifold. Subsequently, we analyze the Jacobi linearization of the system evaluated at the equilibria in order to show that the equilibria are asymptotically stable relative to certain affine subspaces. As this approach does not require any monotonicity features of the system, it may also be applied for more complex systems of the same kind such as bi-directional ribosome flows or time-varying template numbers.
\end{abstract}

\section{Introduction}
The ribosome flow model (RFM) as presented in \cite{Reuveni2011} describes the movement of ribosomes along an mRNA template and can be used to study the dynamics of mRNA translation, an important step in the process of protein synthesis. In order to better understand and eventually design genetic regulatory networks (GRNs), the RFM with pool can be used to provide a simulation framework which not only simulates the production of certain proteins but also takes into account the allocation of the ribosomes. Thus, a better evaluation of the performance of artificial GRNs, an important topic in the field of synthetic biology (\cite{Ceroni2015}), will be possible.\\
Several versions of the RFM have been proposed recently, all of them built upon the works \cite{MacDonald1968} and \cite{Heinrich1980} where a probabilistic model of a growth center moving along a nucleic acid template is considered. This model, also known as the totally asymmetric exclusion process (TASEP) was then simplified and adapted in \cite{Reuveni2011} to obtain the RFM, a deterministic description of the movement of several ribosomes along a strand of mRNA. For this classical RFM, where the amount of ribosomes is assumed to be abundant and only a single mRNA is studied, several results on model properties such as uniqueness of the steady state and convergence to this equilibrium point were presented in \cite{Margaliot2012} wherein the authors make use of the theory of monotone systems and the contraction principle. Recently, \cite{Raveh2015} studied a network of several mRNA templates in interaction with a finite pool of ribosomes, following up the argumentation of \cite{Margaliot2012}, the authors show that the solution of the RFM network with pool monotonically converges to a certain equilibrium point which is determined solely by the initial condition.\\
In contrast to the monotone systems theory approach, we propose a geometric approach to study the stability properties of the model in \cite{Raveh2015} in order to provide a concept to study non-monotonic variants of the RFM in future.\\
In particular, after introducing the detailed system description of the RFM with pool, we show that the equilibria of the RFM with pool constitute a normally hyperbolic invariant submanifold and therefore the restriction of the linearization of this system to the normal spaces of this submanifold can be used to study the stability properties of the submanifold. Finally, we conclude that these equilibria are asymptotically stable relative to certain affine subspaces.

\section{The ribosome flow model with pool}
As described in \cite{Reuveni2011}, the mechanism of translation can be approximated as an initiation event followed by several elongation steps. To be more precise, after binding to the mRNA, the ribosomes perform a unidirectional movement along the mRNA until they reach its end and subsequently unbind. In general, the speed of this motion is not constant, but to focus on the system theoretic analysis and for the sake of simplicity however, a constant elongation speed is assumed in the remainder.\\
The RFM with pool can be modeled as the differential equation
\begin{align}
\dot{R} &= -\lambda R (1-x_1)+\lambda_c  x_n \label{eqn:sys_first}\\
\begin{bmatrix}
\dot{x}_1 \\
\dot{x}_2 \\
\vdots\\
\dot{x}_n 
\end{bmatrix} &= \begin{bmatrix}
\lambda R (1-x_1) - \lambda_c x_1 (1-x_2) \\
\lambda_c x_1 (1-x_2) - \lambda_c x_2 (1-x_3)\\
\vdots\\
\lambda_c x_{n-1} (1-x_n) - \lambda_c x_n
\end{bmatrix}, 
\end{align}
with initial conditions
\begin{align}
R(t=0) &=  R_{\text{tot}} \in [0,\infty)\\
x_i(t=0) &= 0,\quad i=1,\ldots,n.
\end{align} \label{eqn:sys_last}
The model states and parameters are explained in Table (\ref{tab:statesandparam}).
\begin{table}
\renewcommand{\arraystretch}{1.3}
\caption{States and parameters of the RFM with pool.}\label{tab:statesandparam}
\centering
\begin{tabular}{l p{6.8cm}}
$R \in [0,\infty)$:& molecular amount of free ribosomes \\
$x_{i}  \in [0,1]$:& avg. ribosome density at mRNA location $i = 1,\ldots, n$\\
$\lambda \in \mathbb{R}^+$:& initiation rate\\
$\lambda_c  \in \mathbb{R}^+$:& elongation rate\\
$n \in \mathbb{N}$:& number of discretization points on mRNA template
\end{tabular} 
\end{table}
The number of discretization points is determined such that each state can take up exactly one ribosome, therefore $ n = \frac{ml}{rl_{\text{spec}}}$ with $ml$ the length of the mRNA template and $rl_{\text{spec}}$ the specific length a single ribosome occupies on the mRNA template. Typically, the elongation rate $\lambda_c$ is cell type dependent and therefore might be determined from biological data. The remaining free variables are $ml$, $\lambda$ as well as $R_{\text{tot}}$, the total amount of available ribosomes.\\
For simpler notation, we collect all states such that
\begin{align}
\mathbf{x} &= \begin{bmatrix}
R & x_1 & x_2 & \hdots & x_n
\end{bmatrix}^{\top}\\
\dot{\mathbf{x}} &= f(\mathbf{x}) \label{eqn:sys_x}
\end{align}
with $x_0 = R$ and $\mathbf{x} \in \mathbb{R}^{n+1}$. Further, as $x_1, \ldots, x_n$ are densities, it only makes sense to consider solutions 
\begin{align*}
t &\mapsto \mathbf{x}(t) \quad \text{for which}\\
\forall t &\in [0,\infty), \mathbf{x}(t) \in \Omega := [0,\infty) \times [0,1] \times \ldots \times [0,1] 
\end{align*}
and as shown in \cite{Raveh2015}, all solutions starting in $\Omega$ will not only stay in this set but also are separated from the boundary $\partial \Omega$ in finite time.

\section{Model properties}
Similarly to the work of \cite{Raveh2015} and \cite{Margaliot2012}, we are interested in the number of equilibra of System (\ref{eqn:sys_x}) as well as their stability properties. In order to analyze similar models with slightly different conditions such as bidirectional flow on the template (\cite{Edri2014}) or even combinations of such systems, a geometric point of view for the stability analysis will turn out to be beneficial.

\subsection{Equilibrium points}
In order to find all equilibra of the RFM with pool, we bring System (\ref{eqn:sys_x}) into the form
\begin{align}
\dot{\mathbf{x}} &= A(\mathbf{x}) \mathbf{x}
\end{align}
with
\begin{equation}
\begin{split}
&A(\mathbf{x}) = \\ & \resizebox{.44 \textwidth}{!}
{$\begin{bmatrix}
-\lambda (1-x_1) & 0 & \cdots & 0 & \lambda_c \\
\lambda (1-x_1) & -\lambda_c (1-x_2) & 0 & \cdots & 0 \\
0 & \lambda_c (1-x_2) & -\lambda_c (1-x_3) &0 & \vdots \\
\vdots & 0 & \ddots &  \ddots& 0\\
0& \cdots&  0 & \lambda_c (1-x_{n}) & -\lambda_c
\end{bmatrix}
$}
\end{split}
\end{equation}
and rewrite the nullspace of $A$ as
\begin{align}
\ker A(\mathbf{x}) &= \Span (v(\mathbf{x}))\label{eqn:nullspace}
\end{align}
with
\begin{equation}
v(\mathbf{x})=\begin{bmatrix}
\frac{\lambda_c}{\lambda (1-x_1)} & \frac{1}{1-x_2} & \cdots & \frac{1}{1-x_{n}} & 1
\end{bmatrix}^{\top}. 
\end{equation}
This curve now represents the continuum of all equilibria of (\ref{eqn:sys_x}) through $f(\mathbf{x})=0 \Leftrightarrow \mathbf{x} \in \Span (v(\mathbf{x}))$. However, notation (\ref{eqn:nullspace}) is quite unhandy and for that reason we will first give an alternative representation of (\ref{eqn:nullspace}).\\
Instead of using the span of a state dependent vector we can define a parameterization $\gamma: s \mapsto \gamma(s)$ of the nullspace with $s$ a scalar independent variable. Specifically, we define $\gamma$ such that
\begin{align}
\begin{split}
\forall \mathbf{x} \in \text{int }\Omega:& \quad \mathbf{x} \in \ker A(\mathbf{x}) \\
\exists s > 0:& \quad \mathbf{x} =\gamma(s). \label{eqn:gammadef}
\end{split}
\end{align}
Calculating $\gamma$ is straight forward and can be achieved by multiplying $v$ with the independent variable $s$ and then recursively solving $\gamma(s) = s v(\mathbf{x})$ and substituting all $x_i$, which then results in 
\begin{equation}
\gamma: s \mapsto \begin{bmatrix}
\gamma_0(s) & \ldots & \gamma_n(s) 
\end{bmatrix}^{\top}
\end{equation}
with the components $\gamma_i(s)$ given recursively as a series of continued fractions with
\begin{equation}
\gamma_i(s) = \begin{cases}
\tfrac{\lambda_c s}{\lambda(1-\gamma_{1}(s))} & i = 0\\
 \tfrac{s}{1-\gamma_{i+1}(s)} & i = 1 \ldots (n-1)\\
 s & i=n.
 \end{cases} \label{eqn:gamma_recu}
\end{equation}
In the remainder, we restrict our attention to solutions initialized in int $\Omega$, the interior of $\Omega$. This is justified as solutions initialized on $\partial \Omega$ attain values in int $\Omega$ in finite time. Under this assumption we notice that 
\begin{equation}
\exists \bar{s} > 0: s\in (0,\bar{s}) \Rightarrow \gamma(s) \in \text{int } \Omega. \label{eqn:gamma_interval}
\end{equation}
We henceforth restrict the domain of $\gamma$ to $(0,\bar{s})$. This is a rather technical assumption to ensure that we study the RFM with pool in a domain which makes sense biologically. It is further possible to show that for $s^*>\bar{s}$, $\gamma(s^*) \notin \Omega$ and for this case, the model ceases to have any biological meaning. The value of $\bar{s}$ is only dependent on the number of discretization points $n$ and further is a solution of the polynomial equation
\begin{equation}
1+\sum_{j=0}^{\left\lfloor \tfrac{n-1}{2} \right\rfloor} (-\bar{s})^{j+1} \binom{n-j}{j+1}=0. \label{eqn:sbar}
\end{equation} 
The derivation of this result is given in the appendix and we note on the side that this $\bar{s}$ can be used to calculate an upper bound on the protein production rate $\kappa=\lambda_c \bar{s} $ of the studied mRNA template.\\
With this representation of the equilibria of (\ref{eqn:sys_x}) at hand, we can proceed with studying their stability properties.

\subsection{Stability of equilibria}
In this section, we consider the Jacobian of $f$ in order to study the stability properties of the equilibria:
\begin{align}
J_f(\mathbf{x})&=\begin{bmatrix}
\frac{\partial f_0}{\partial x_0}(\mathbf{x}) &  \cdots & \frac{\partial f_0}{\partial x_n}(\mathbf{x}) \\
\vdots & \ddots &  \vdots \\
\frac{\partial f_n}{\partial x_0}(\mathbf{x}) &\cdots & \frac{\partial f_n}{\partial x_n}(\mathbf{x}) 
\end{bmatrix}\\
&= A(\mathbf{x})+ \begin{bmatrix}
0 & \lambda x_0 & 0 & \cdots & 0 \\
0 & -\lambda x_0 & \lambda_c x_1 &\ddots & \vdots \\
0 & \ddots & -\lambda_c x_1 & \ddots & 0 \\
\vdots &  &  &  \ddots& \lambda_c x_{n-1}\\
0& \cdots&  & 0 & -\lambda_c x_{n-1}
\end{bmatrix}
\end{align}
and evaluate this Jacobian at $\gamma(s)$, i.e.
\begin{equation}
\begin{split}
& J_f(\gamma)= \\ & \resizebox{.43 \textwidth}{!}
{$
\begin{bmatrix}
-\lambda(1-\gamma_1) & \lambda \gamma_0                        & 0                                & \cdots   & \lambda_c \\
\lambda(1-\gamma_1)  & -\lambda_c(1-\gamma_2)-\lambda \gamma_0 & \lambda_c \gamma_1               &    &   0\\
0                    & \lambda_c(1-\gamma_2)                   & -\lambda_c (1-\gamma_3+\gamma_1) &   &  \vdots \\
\vdots               & 0                                      &  \lambda_c (1-\gamma_3)          &  \ddots  & \lambda_c \gamma_{n-1}   \\
0                     &                  \cdots                      &            0                      &   \lambda_c(1-\gamma_n) & -\lambda_c (1+\gamma_{n-1})
\end{bmatrix}\label{eqn:jacobiKer}
$}
\end{split}
\end{equation}
where we omitted the argument $s$ for the sake of readability.
\begin{thm} \label{thm: EV}
For all $s>0$ the Jacobian matrix of $f$, evaluated at $\gamma(s)$, has an eigenvalue equal to $0$. Further, all remaining eigenvalues have real parts strictly smaller than zero.
\end{thm}
\begin{proof}
In (\ref{eqn:jacobiKer}) one can see that all diagonal elements of $J_f(\gamma)$ are strictly negative as $\gamma_i<1$ for $i = 1,\ldots,n$ and further, with $J_{f}^{k,i}(\gamma)$ being the element of $J_f(\gamma)$ in the $k$-th row and $i$-th column,
\begin{equation}
\sum_{k\neq i} \vert J_{f}^{k,i}(\gamma) \vert = - J_{f}^{i,i}(\gamma). \label{eqn: GerschRad}
\end{equation}This means that, using Gerschgorin circles (\cite{Gerschgorin1931}) with their center coordinates given by the value of the diagonal elements $J_{f}^{i,i}(\gamma)$ and their radius given by $\sum_{k\neq i} \vert J_{f}^{k,i}(\gamma)\vert$, all eigenvalues have a real part smaller or equal to zero.\\
It remains to show that $J_f(\gamma)$ has precisely one eigenvalue equal to zero which then also implies that the remaining eigenvalues have a real part strictly smaller than zero. Therefore we use the theorem on the reduced row echelon form and bring $J_f(\gamma)$ into upper triangular form, such that
\begin{align}
J_f(\gamma) &= L U, \quad \text{with }\\
L &=\begin{bmatrix}
1 & 0 & & \cdots & 0\\
-1& 1 &  &  & \\
0 & \ddots & \ddots &  & \vdots\\
\vdots&  & -1 & 1 & 0\\
0 & \cdots & 0 & -1 & 1
\end{bmatrix}\\
U &= \notag \\& \hspace{-1.1cm}\resizebox{.44 \textwidth}{!}
{$
\begin{bmatrix}
-\lambda(1-\gamma_1) & \lambda \gamma_0                        & 0                                & \cdots   & 0 & \lambda_c \\
0  & -\lambda_c(1-\gamma_2) & \lambda_c \gamma_1               &    &										 &   \lambda_c\\
\vdots        &       & -\lambda_c (1-\gamma_3) &  \ddots  & 														&  \vdots \\
               &          \ddots                            &       		 & \ddots  &   \lambda_c \gamma_{n-2} &  \lambda_c  \\
               &                                      &        &   &  -\lambda_c (1-\gamma_{n})  & \lambda_c (1+\gamma_{n-1})   \\
0                     &                  \cdots                      &                               &   &  0  &0
\end{bmatrix}
$}.
\end{align}
Now that the rank of $U$ is $n-1$, this concludes the proof.
\end{proof}
One might be tempted to use Lyapunov's indirect method (\cite{Khalil2002}), that asymptotic stability of a hyperbolic equilibrium is determined by the Jacobi linearization, or the theorem of Hartman-Grobman (\cite{Hartman1960}), that a vector field and its linearization are conjugate in a neighborhood of a hyperbolic equilibrium, in order to draw conclusions about the stability properties of the equilibria. However, the existence of the zero eigenvalue means that the equilibria are non-hyperbolic and these methods are not applicable. Yet, \cite{Pugh1970} offers an extension to normally hyperbolic invariant manifolds in the following sense: a vector field and the restriction of its linearization to the normal spaces of a given normally hyperbolic invariant manifold are conjugate in a neighborhood of the normally hyperbolic invariant manifold. Therefore, following the notation of \cite{Hirsch1977}, we will show in the next section that $f$ is normally hyperbolic at $\gamma((0,\bar{s}))$ in order to continue with the stability analysis.

\subsection{Normal hyperbolicity}
In this section, we proceed as follows: we already noted that $\gamma((0,\bar{s}))$ is a manifold of equilibria (and thus invariant) and that $J_f(\gamma)$, the Jacobian of $f$ evaluated on $\gamma$, has precisely one zero eigenvalue. Next, we show that the eigenvector associated with the zero eigenvalue lies in the tangent space $T_{\gamma(s)}\gamma((0,\bar{s}))=\Span (\lbrace \dot{\gamma}(s)\rbrace)$ of $\gamma((0,\bar{s}))$ at $\gamma(s)$, no matter at which $s \in (0,\bar{s})$ we evaluate $\gamma$. Further, as shown in Theorem (\ref{thm: EV}), all other eigenvalues of the Jacobian are negative if evaluated on $\gamma((0,\bar{s}))$ and we will show that their eigenvectors span $N_{\gamma(s)}\gamma((0,\bar{s}))$, the normal space of $\gamma((0,\bar{s}))$ at $\gamma(s)$ in $\mathbb{R}^{n+1}$. In conclusion, $\gamma((0,\bar{s}))$ is not only normally hyperbolic, but further asymptotically stable in a sense we detail further below.

\begin{lem}\label{lem: tangential space}
The eigenvector associated with the zero eigenvalue of $J_f(\gamma(s))$ is linearly dependent on $\dot{\gamma}(s)=\frac{d}{ds}\gamma(s)$. 
\end{lem}
\begin{proof}
It is sufficient to show that
\begin{equation}
J_f(\gamma) \dot{\gamma} = 0.
\end{equation}
We thus consecutively show that
\begin{align}
J_{f}^0(\gamma) \dot{\gamma} &= 0 \label{eqn: thm2_pt1}\\
J_{f}^i(\gamma) \dot{\gamma} &= 0 \quad i = 1,\ldots,n-1 \label{eqn: thm2_pt2}\\
J_{f}^n(\gamma) \dot{\gamma} &= 0 \label{eqn: thm2_pt3}
\end{align}
with $J_{f}^i$ the $i+1$-th row of the Jacobian $J_{f}$. We start with showing (\ref{eqn: thm2_pt3}):
\begin{align}
J_{f}^n(\gamma) \dot{\gamma} &= \lambda_c(1-\gamma_n)\dot{\gamma}_{n-1} -\lambda_c(1+\gamma_{n-1})\dot{\gamma}_{n}\\
&=\lambda_c(1-s)\frac{1}{(s-1)^2} -\lambda_c(1+\frac{s}{1-s})\\
&= \lambda_c \left( - \frac{1}{s-1} + \frac{1}{s-1} \right) =0.
\end{align}
Showing (\ref{eqn: thm2_pt2}) can now be achieved in a general form since for all $i = 1,\ldots,n-1 $ the structure of $J_{f}^i$ is identical, namely
\begin{equation}
\begin{split}
J_{f}^i(\gamma) \dot{\gamma} =& \lambda_c (1-\gamma_i) \dot{\gamma}_{i-1}\\
& - \lambda_c (1-\gamma_{i+1}+\gamma_{i-1}) \dot{\gamma}_i\\
& + \lambda_c \gamma_i \dot{\gamma}_{i+1}.
\end{split}
\end{equation} 
We rearrange the last equation to get
\begin{align}
\begin{split}
J_{f}^i(\gamma) \dot{\gamma} =& \lambda_c \left( \dot{\gamma}_{i-1} -  \gamma_{i}\dot{\gamma}_{i-1}-\gamma_{i-1}\dot{\gamma}_{i}  \right)\\
&- \lambda_c \left( \dot{\gamma}_{i} -  \gamma_{i+1}\dot{\gamma}_{i}-\gamma_{i}\dot{\gamma}_{i+1}  \right)
\end{split}\\
\begin{split}
=& \lambda_c \left( \dot{\gamma}_{i-1} - \dot{\overline{(\gamma_{i-1}\gamma_{i})}} \right)\\
&- \lambda_c \left( \dot{\gamma}_{i} - \dot{\overline{(\gamma_{i}\gamma_{i+1})}}\right).
\end{split}
\end{align}
Now, we can utilize the generating equation (\ref{eqn:gamma_recu}) again to realize that
\begin{align}
\dot{\gamma}_i - \dot{\overline{(\gamma_i \gamma_{i+1})}} &=\begin{cases} \frac{\lambda_c}{\lambda} & i =0\\1 & i =1,\ldots,n-1 \end{cases} \label{eqn:gammadot}
\end{align}
and using (\ref{eqn:gammadot}) for $i =1,\ldots,n-1$ to arrive at the equality
\begin{equation}
J_{f}^i(\gamma) \dot{\gamma} =0 \quad i =1,\ldots,n-1.
\end{equation}
Finally, we merely need to verify whether this is also true for (\ref{eqn: thm2_pt1}):
\begin{align}
J_{f}^0(\gamma) \dot{\gamma} &= -\lambda(1-\gamma_1)\dot{\gamma}_{0} + \lambda \gamma_{0} \dot{\gamma}_1 + \lambda_c \dot{\gamma}_n\\
&= \lambda (-\dot{\gamma}_{0} +\dot{\gamma}_{0} \gamma_1 +\dot{\gamma}_{1} \gamma_0 ) + \lambda_c \dot{\gamma}_{n}.
\end{align}
Now with $\dot{\gamma}_n=1$, and using equation (\ref{eqn:gammadot}) for $i=0$,
\begin{align}
J_{f}^0(\gamma) \dot{\gamma} &= \lambda \left(-\frac{\lambda_c}{\lambda}\right)+ \lambda_c =0.
\end{align}
This concludes the proof. 
\end{proof}

One says that $f$ is normally hyperbolic at $\gamma((0,\bar{s}))$ if the derivative of $f$ evaluated at $\gamma(s)$ leaves the continuous splitting
\begin{equation}
\mathbb{R}^{n+1} = N^u \oplus T_{\gamma(s)}\gamma((0,\bar{s})) \oplus N^s
\end{equation}
invariant and if the normal behavior dominates the tangent one. $N^u$ and $N^s$ in that respect are the subspaces spanned by the normal eigenvectors of the Jacobian associated with positive and negative eigenvalues respectively. Due to Theorem (\ref{thm: EV}) we know that $N^u=\emptyset$. Lemma (\ref{lem: tangential space}) shows that the dynamics of $f$ on $T_{\gamma(s)}\gamma((0,\bar{s}))$ is determined by the zero eigenvalue and it remains to show that the (generalized) eigenvectors of the remaining eigenvalues span the normal space of $\gamma((0,\bar{s}))$ at any $\gamma(s)$. In order to do so, we define the affine subspaces 
\begin{align}
S_p:= \left\lbrace \mathbf{x} \middle| \sum_{i=0}^{n} x_i = p \right\rbrace
\end{align} 
and note the following proposition:
\begin{prop}\label{prop:1}
Any solution of System (\ref{eqn:sys_x}) initialized on a certain $S_p$ will stay on $S_p$ for all times and $p = R_{\text{tot}}$.
\end{prop}
\begin{proof}
Assume the initial state of System (\ref{eqn:sys_x}) lies in $S_p$, i.e.  $p = R_{\text{tot}}$, then choose 
\begin{align}
V(\mathbf{x})=\sum_{i=0}^{n} x_i
\end{align}
as a Lyapunov function such that $S_p$ is just the level set $V^{-1}(\lbrace p \rbrace)$ of $V$. Now consider the Lie derivative $\mathcal{L}_fV$ of $V$ along $f$,
\begin{align}
\mathcal{L}_fV &= \frac{\partial V}{\partial \mathbf{x}} \cdot f\\
 &= \begin{bmatrix}
 1 & \cdots & 1
 \end{bmatrix} f\\
 &= 0.
\end{align}
This shows that any level set of $V$ is invariant under (\ref{eqn:sys_x}). 
\end{proof}
With Proposition (\ref{prop:1}) at hand, it remains to show that the intersection of $\gamma$ with all $S_p$ is always transversal and never tangential to conclude normal hyperbolicity of $f$ at $\gamma((0,\bar{s}))$.

\begin{lem}\label{lem: transv}
For all $p>0$, the curve $\gamma$ intersects $S_p$ uniquely and transversely.
\end{lem}
In other words, this means that the continuum of equilibria of system (\ref{eqn:sys_x}) represented by the curve $\gamma$ never intersects the $n$-dimensional affine subspace $S_p$ tangentially. \newline
\begin{proof}
We start with showing the transversality of the intersection of $\gamma$ and $S_p$ as the uniqueness of this intersection then follows as we will show at the end of the proof. In order to do so, we note that the affine subspace $S_p$ is always the same subspace $M$ translated in direction of $x_0$
\begin{equation}
S_p= \{e_1 p\} + M 
\end{equation}
with $e_1$ the canonical unit vector and the subspace $M$ defined as the image of
\begin{align}
\Lambda &=\begin{bmatrix}
\mu_1&\mu_2&\ldots&\mu_N
\end{bmatrix}=\begin{bmatrix}
-1 & -1 & \cdots & -1 \\
1 & 0 & \cdots & 0 \\
0 & 1 &  &  \vdots \\
\vdots & 0 &  & 0\\
0 & 0 & \cdots & 1 
\end{bmatrix}.
\end{align}
This means that the only vectors perpendicular to $M$ are multiples of the all ones vector $\mathds{1}$ and it is thus sufficient to show that the velocity vector of $\gamma$ never is perpendicular to $\mathds{1}$, in other words
\begin{equation}
\langle \dot{\gamma}, \mathds{1}\rangle \neq 0.
\end{equation}
Showing this will be achieved by noting that $\dot{\gamma}_i$ is positive for all $i=0,\ldots,n$. This can be achieved by the following inductive argument starting with the highest appearing index $n$ and then reducing it step wise.\\
First, we consider $\gamma_n(s)=s$ and see that
\begin{equation}
\dot{\gamma}_n(s)=1 >0.\label{eqn:ind_start}
\end{equation}
Next, we assume that the statement that 
\begin{equation}
\dot{\gamma}_i > 0 \label{eqn:ind_assum}
\end{equation} is true. Now further reducing the index, we need to show that (\ref{eqn:ind_assum}) still holds for $\dot{\gamma}_{i-1}$. Therefore we differentiate the generating equation (\ref{eqn:gamma_recu}) for index $i-1$ to arrive at
\begin{equation}
\dot{\gamma}_{i-1}(s)=\frac{(1-\gamma_{i})+s\dot{\gamma}_i}{(1-\gamma_{i})^2}.
\end{equation}
As shown in (\ref{eqn:gamma_interval}) we already know that $\gamma_i \in (0,1) \quad \forall i>0$. Together with (\ref{eqn:ind_assum}), this means that 
\begin{equation}
\dot{\gamma}_{i-1}(s) > 0
\end{equation} 
which concludes the proof for the transversality of the intersection. Now that we further know that all derivatives of $\gamma$ with respect to $s$ are larger than zero for $s \in (0,\bar{s})$, the uniqueness of the intersection follows directly from the combination of these arguments. 
\end{proof}
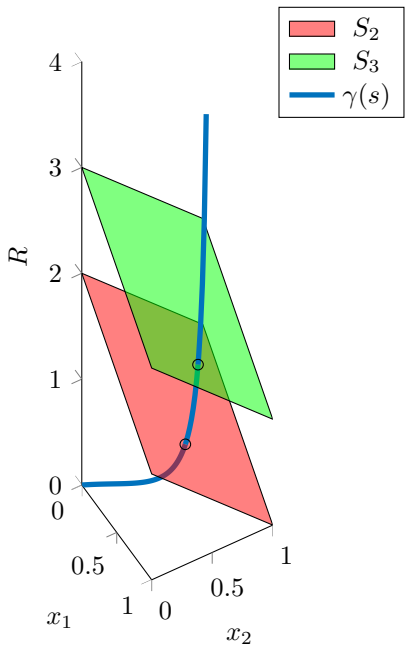
\begin{figure}
\centering
%
%
\definecolor{mycolor1}{rgb}{0.00000,0.44700,0.74100}%
\definecolor{mycolor2}{rgb}{0.85000,0.32500,0.09800}%
\begin{tikzpicture}

\begin{axis}[%
width=1in,
height=3in,
scale only axis,
plot box ratio=1 1 4,
xmin=0,
xmax=1,
restrict x to domain=0:1.6,
restrict y to domain=0:1.6,
restrict z to domain=0:4.1,
tick align=outside,
xlabel={$x_1$},
ymin=0,
ymax=1,
ylabel={$x_2$},
zmin=0,
zmax=4,
zlabel={$R$},
view={60}{46},
axis background/.style={fill=white},
axis x line*=bottom,
axis y line*=left,
axis z line*=left,
legend pos=outer north east
]
\addplot3 [color=mycolor1,solid,line width=2.0pt,fill opacity=0.8,forget plot]
 table[row sep=crcr] {%
0	0	0\\
0.00507512175114543	0.00504949494949495	0.00253762623685843\\
0.0102020199958776	0.0100989898989899	0.00510154097250625\\
0.0153814911338734	0.0151484848484848	0.00769256555309408\\
0.0206143479820702	0.0201979797979798	0.0103115559005606\\
0.0259014201998953	0.0252474747474748	0.0129594043513829\\
0.031243554727783	0.0302969696969697	0.0156370416139303\\
0.0366416162394621	0.0353464646464646	0.0183454388534448\\
0.0420964876085254	0.0403959595959596	0.0210856099144622\\
0.0476090703898058	0.0454454545454545	0.0238586136913625\\
0.0531802853161137	0.0504949494949495	0.0266655566586973\\
0.0588110728109081	0.0555444444444444	0.0295075955739995\\
0.064502393517503	0.0605939393939394	0.0323859403669533\\
0.0702552288454324	0.0656434343434344	0.0353018572300856\\
0.0760705815346288	0.0706929292929293	0.0382566719275748\\
0.0819494762380945	0.0757424242424242	0.0412517733403461\\
0.087892960123778	0.0807919191919192	0.0442886172673786\\
0.0939021034963984	0.0858414141414141	0.0473687305050889\\
0.0999780004399912	0.0908909090909091	0.0504937152288181\\
0.106121769481989	0.095940404040404	0.0536652537028481\\
0.112334554279679	0.100989898989899	0.0568851133480519\\
0.118617524329931	0.106039393939394	0.0601551521992639\\
0.124971875703107	0.111088888888889	0.0634773247877899\\
0.131398831802135	0.116138383838384	0.0668536884882061\\
0.137899644147748	0.121187878787879	0.0702864103727665\\
0.144475593190948	0.126237373737374	0.0737777746214254\\
0.151127989153816	0.131286868686869	0.0773301905407375\\
0.157858172899803	0.136336363636364	0.0809462012508153\\
0.16466751683474	0.141385858585859	0.0846284931061918\\
0.171557425839823	0.146435353535354	0.0883799059239565\\
0.17852933823792	0.151484848484848	0.09220344410104\\
0.185584726794574	0.156534343434343	0.0961022887121492\\
0.19272509975519	0.161583838383838	0.10007981069078\\
0.199952001919923	0.166633333333333	0.104139585208146\\
0.207267015757883	0.171682828282828	0.108285407378983\\
0.214671762562344	0.176732323232323	0.112521309439324\\
0.222167903648727	0.181781818181818	0.116851579559739\\
0.229757141597219	0.186831313131313	0.121280782478618\\
0.237441221541977	0.191880808080808	0.125813782164315\\
0.245221932508971	0.196930303030303	0.130455766742748\\
0.253101108804618	0.201979797979798	0.135212275959158\\
0.261080631457465	0.207029292929293	0.140089231479778\\
0.269162429715323	0.212078787878788	0.145092970382036\\
0.27734848260032	0.217128282828283	0.150230282231729\\
0.285640820524541	0.222177777777778	0.155508450203523\\
0.294041526969002	0.227227272727273	0.16093529676872\\
0.302552740228881	0.232276767676768	0.166519234553301\\
0.311176655228085	0.237326262626263	0.172269323061958\\
0.319915525406374	0.242375757575758	0.178195332072965\\
0.328771664682457	0.247425252525253	0.18430781263741\\
0.337747449496656	0.252474747474747	0.190618176768707\\
0.346845320936915	0.257524242424242	0.197138787089183\\
0.356067786952162	0.262573737373737	0.203883057916091\\
0.365417424657232	0.267623232323232	0.210865569527083\\
0.374896882733795	0.272672727272727	0.218102197654381\\
0.384508883932005	0.277722222222222	0.225610260629287\\
0.394256227677817	0.282771717171717	0.233408687049049\\
0.404141792791221	0.287821212121212	0.241518207384835\\
0.414168540320945	0.292870707070707	0.249961573616373\\
0.424339516501477	0.297920202020202	0.25876381179547\\
0.434657855838623	0.302969696969697	0.267952513445746\\
0.445126784330155	0.308019191919192	0.277558172948887\\
0.455749622828507	0.313068686868687	0.287614579612904\\
0.466529790552882	0.318118181818182	0.298159275049187\\
0.477470808758566	0.323167676767677	0.309234088912707\\
0.488576304571734	0.328217171717172	0.320885769129569\\
0.4998500149985	0.333266666666667	0.333166726642676\\
0.51129579111755	0.338316161616162	0.346135919710831\\
0.522917602466204	0.343365656565657	0.359859909253236\\
0.53471954163043	0.348415151515152	0.374414125123655\\
0.546705829049948	0.353464646464646	0.389884394193538\\
0.558880818050265	0.358514141414141	0.406368795650099\\
0.571249000114273	0.363563636363636	0.423979928280675\\
0.583815010406784	0.368613131313131	0.442847700578315\\
0.596583633566316	0.373662626262626	0.463122789942698\\
0.609559809779295	0.378712121212121	0.48498096596824\\
0.622748641152906	0.383761616161616	0.508628540576259\\
0.63615539840384	0.388811111111111	0.534309303209975\\
0.649785527881375	0.393860606060606	0.562313435646939\\
0.663644658944443	0.398910101010101	0.592989098609573\\
0.677738611713666	0.403959595959596	0.626757673495578\\
0.692073405220816	0.409009090909091	0.664134079101536\\
0.706655265979637	0.414058585858586	0.705754250611234\\
0.721490637003702	0.419108080808081	0.752412910465872\\
0.736586187298714	0.424157575757576	0.805116427661625\\
0.7431 				0.4263  			0.8296\\
};

\addplot3[area legend,solid,draw=black,fill=red,fill opacity=0.5]
table[row sep=crcr] {%
x	y	z\\
0	0	2\\
1	0	1\\
1	1	0\\
0	1	1\\
}--cycle;
\addlegendentry{$S_2$};

\addplot3 [color=mycolor1,solid,line width=2.0pt,fill opacity=0.8,forget plot]
 table[row sep=crcr] {%
0.7431 				0.4263  			0.8296\\
0.751948821858658	0.429207070707071	0.865158299031551\\
0.76758569219727	0.434256565656566	0.934229415052097\\
0.783504207448562	0.439306060606061	1.0145833677153\\
0.799712046072628	0.444355555555556	1.10929176428825\\
0.816217168205559	0.449405050505051	1.22265242655447\\
0.833027828695217	0.454454545454545	1.36086912538559\\
0.850152590867722	0.45950404040404	1.53323985734852\\
0.8614				0.4628				1.669\\
};

\addplot3[area legend,solid,draw=black,fill=green,fill opacity=0.5]
table[row sep=crcr] {%
x	y	z\\
0	0	3\\
1	0	2\\
1	1	1\\
0	1	2\\
}--cycle;
\addlegendentry{$S_3$};

\addplot3 [color=mycolor1,solid,line width=2.0pt,fill opacity=0.8]
 table[row sep=crcr] {%
0.8614				0.4628				1.669\\
0.86760034107287	0.464553535353535	1.75436077070718\\
0.885380304060424	0.46960303030303	2.04852676694671\\
0.903502059243137	0.474652525252525	2.45939199080146\\
0.921975557906794	0.47970202020202	3.07404966528938\\
0.940811141432201	0.484751515151515	4.09495576432047\\
0.960019560599251	0.48980101010101	6.1255080914873\\
0.979611996048774	0.494850505050505	12.1358252194364\\
0.999600079984003	0.4999	624.999975\\
};
\addlegendentry{$\gamma(s)$}; 

\addplot3 [color=black,only marks,mark=o, forget plot]
table[row sep=crcr] {%
0.8614	0.4628	1.669\\
0.7431  0.4263  0.8296\\
};

\end{axis}
\end{tikzpicture}%
 \caption{Two affine subspaces $S_2$ and $S_3$ and the nullspace of $A(\mathbf{x})$ depicted for the first three dimensions.} \label{fig:1}
\end{figure}

Fig. (\ref{fig:1}) illustrates two affine subspaces $S_2$ (red) and $S_3$ (green) and the continuum of equilibria $\gamma((0,\frac{1}{2}))$ (blue) in the first three coordinates. One can see that the subspaces are just shifted by the difference in total amount of ribosomes $R_{\text{tot}}$ in $R$ direction and that the curve intersects with each subspace uniquely and not tangentially. We are now able to formulate our main result.

\begin{thm}
The invariant set $\gamma((0,\bar{s}))$ of (\ref{eqn:sys_x}) is asymptotically stable.
\end{thm}

\begin{proof}
With Lemmata (\ref{lem: tangential space}) and (\ref{lem: transv}) at hand we can conclude that $f$ is normally hyperbolic at $\gamma((0,\bar{s}))$ and therefore $f$ and the restriction of its linearization to the normal spaces of $\gamma((0,\bar{s}))$ are conjugate in a neighborhood of $\gamma((0,\bar{s}))$. As shown in Theorem (\ref{thm: EV}), the eigenvalues of this linearization evaluated at $\gamma(s)$ are all strictly negative except for one, which is exactly zero. Now that the eigenvector associated with this zero eigenvalue lies in the tangent space $T_{\gamma(s)}\gamma((0,\bar{s}))$ and we are only interested in the eigenvectors lying in the normal spaces of $\gamma((0,\bar{s}))$, which thus are all associated with strictly negative eigenvalues, this concludes the proof. 
\end{proof}
We further showed in Proposition (\ref{prop:1}) that the affine subspaces $S_p$ are also invariant under the system dynamics. This further means that any solution initialized on a certain $S_p$ with $p>0$ will converge to the unique equilibrium given by the intersection of $\gamma(s)$ with $S_p$. In other words, this intersection point is asymptotically stable relative to $S_p$, following the terminology of \cite{Bhatia1970}.

\section{Conclusion}
We considered the RFM with pool, a model describing the movement of ribosomes along a single mRNA template as well as the dynamics of a pool of available ribosomes. We found that the equilibria of the system can be characterized by a curve $\gamma$ and that there exist affine subspaces $S_p$ which are invariant under the dynamics of the system. In order to characterize the stability of the equilibria we studied the Jacobi linearization of the system evaluated on $\gamma$ and found that all eigenvalues are smaller than zero except for precisely one which is equal to zero, therefore the equilibria are non-hyperbolic. In order to draw any conclusions from the linearization we showed that the system under study is normally hyperbolic at $\gamma$. This was achieved in two steps, first, showing that the eigenvector associated with the zero eigenvalue of the Jacobian evaluated on $\gamma$ is linearly dependent on the velocity vector of $\gamma$ and second, showing that $\gamma$ intersects all $S_p$ transversely. This insight then enabled us to apply the results of \cite{Pugh1970} in order to conclude that the equilibria of the system are asymptotically stable relative to the affine subspaces $S_p$.\\
In previous works on the RFM (\cite{Raveh2015}) monotone systems theory was used to show that every equilibrium point is semistable in a sense that any solution initialized on a certain $S_p$ monotonically converges to a unique equilibrium point which is dependent on the initial condition. Our results now offer a more detailed characterization of the stability of the RFM with pool and further introduce a geometric approach for studying similar systems with higher complexity as this approach does not require any monotonicity features of the system. Such more complex systems may for instance be models where the copy number of templates is varying over time or the flow of ribosomes is allowed to be bi-directional.

\addtolength{\textheight}{-3.2cm}   
                                  

\bibliography{synthbio}

\section*{Appendix}
In this section, we derive the formula for calculating $\bar{s}$, given by equation (\ref{eqn:sbar}). Therefore, we use a result from \cite{Dudley1987} on continued fractions to find an alternative notation for the components of $\gamma$:
\begin{equation}
\gamma_i(s) = \begin{cases}
 \tfrac{\lambda_c  p_{0}(s)}{\lambda  q_{0}(s)}& i = 0\\
 \tfrac{p_{i}(s)}{q_{i}(s)}& i = 1 \ldots n
 \end{cases} 
\end{equation}
with
\begin{align}
p_i (s)&=p_{i+1}-s p_{i+2} \label{eqn:impl_pi}\\ 
q_i (s)&=q_{i+1}-s q_{i+2} \label{eqn:qi}\\
p_{n+1}&=0\\
p_{n+2}&=-1\\
q_{n+1}&=1\\
q_{n+2}&=0.
\end{align}
We further found that the $i$-th component of a general series of continued fractions of the form
\begin{equation}
r_i (s)=r_{i+1}-s r_{i+2}
\end{equation}
with arbitrary initial factors $r_{n+1}$ and $r_{n+2}$ can be calculated explicitly with the formula
\begin{align}
\begin{split}
r_i(s)=&\sum_{j=0}^{ \left\lfloor \tfrac{n+1-i}{2} \right\rfloor } (-s)^{j} \binom{n+1-i-j}{j}r_{n+1} \\
&+\sum_{j=0}^{\left\lfloor \tfrac{n-i}{2} \right\rfloor} (-s)^{j+1} \binom{n-i-j}{j} r_{n+2} 
\end{split}
\end{align}
where $\left\lfloor \alpha \right\rfloor$ denotes the largest integer smaller than $\alpha$.
Applying this to the problem at hand results in
\begin{align}
p_i(s)&=\sum_{j=0}^{\left\lfloor \tfrac{n-i}{2} \right\rfloor}  -(-s)^{j+1} \binom{n-i-j}{j} \\
q_i(s)&=\sum_{j=0}^{\left\lfloor \tfrac{n+1-i}{2} \right\rfloor}  (-s)^{j} \binom{n+1-i-j}{j} \label{eqn:expl_qi}
\end{align}
and to the best knowledge of the authors, this is the first time an explicit form for this kind of problems has been formulated.\\
From Equation (13) it is possible to deduct that $\gamma_i > \gamma_{i+1}$ for $i=1\ldots n$ as long as $\gamma \in \Omega$ and thus, the first violation of the state constraints occurs when $\gamma_1(\bar{s}) = 1$ in which case $\gamma_0(\bar{s})$ is not defined due to a division by zero. This means that $\bar{s}$ is the smallest positive solution of 
\begin{align}
p_1(\bar{s})&=q_1(\bar{s})\\
\sum_{j=0}^{\left\lfloor \tfrac{n-1}{2} \right\rfloor}  -(-\bar{s})^{j+1} \binom{n-1-j}{j} &= \sum_{j=0}^{\left\lfloor \tfrac{n}{2} \right\rfloor}  (-\bar{s})^{j} \binom{n-j}{j}. \label{eqn:sbar_deduction}
\end{align}
We now need to discriminate between the two cases that $n$ is either even or odd.\\
If $n$ is even, it then follows that
\begin{align}
\left\lfloor \frac{n-1}{2} \right\rfloor &= \frac{n}{2}-1 \\
\left\lfloor \frac{n}{2} \right\rfloor &= \frac{n}{2}
\end{align}
and thus
\begin{align}
\sum_{j=0}^{\tfrac{n}{2}-1}  -(-\bar{s})^{j+1} \binom{n-1-j}{j} &= \sum_{j=0}^{ \tfrac{n}{2} }  (-\bar{s})^{j} \binom{n-j}{j}
\end{align}
which, with a shift of indices and the identity
\begin{equation}
\binom{n-1-j}{j}+\binom{n-1-j}{j+1}=\binom{n-j}{j+1},
\end{equation} 
can be rearranged to
\begin{equation}
1+\sum_{j=0}^{\tfrac{n}{2}-1} (-\bar{s})^{j+1} \binom{n-j}{j+1}=0. \label{eqn:sbar_res1}
\end{equation}
In the other case that $n$ is odd, it follows that
\begin{align}
\left\lfloor \frac{n-1}{2} \right\rfloor =\left\lfloor \frac{n}{2} \right\rfloor = \frac{n-1}{2}
\end{align}
and with the same line of arguments, Equation (\ref{eqn:sbar_deduction}) can be rearranged to
\begin{equation}
1+\sum_{j=0}^{\tfrac{n-1}{2}} (-\bar{s})^{j+1} \binom{n-j}{j+1}=0.\label{eqn:sbar_res2}
\end{equation}
Now, due to the fact that 
\begin{equation}
\left\lfloor \frac{n-1}{2}\right\rfloor = \begin{cases} \frac{n-1}{2} & n \text{ odd} \\ \frac{n}{2}-1 & n \text{ even}
\end{cases}
\end{equation}
Equations (\ref{eqn:sbar_res1}) and (\ref{eqn:sbar_res2}) can be joined to arrive at 
\begin{equation}
1+\sum_{j=0}^{\left\lfloor \tfrac{n-1}{2}\right\rfloor} (-\bar{s})^{j+1} \binom{n-j}{j+1}=0.
\end{equation}
This equation now may have several solutions, however only the smallest positive solution is relevant in order to arrive at an interval $(0,\bar{s})$ for which all $s \in (0,\bar{s})$ yield a $\gamma(s) \in \Omega$.
\end{document}